\documentclass{amsart}

\newcommand{\E}{\mathbb{E}}
\newcommand{\Id}{\mathrm{d}}
\newcommand{\R}{\mathbb{R}}

\usepackage{graphicx,subfig}
\graphicspath{{./Figures/}}
\usepackage[official]{eurosym}
\usepackage{bm}
\usepackage[OT2,OT1]{fontenc}
\usepackage{multirow}
\usepackage{color}
\usepackage[margin=1.2in]{geometry}

\newcommand\cyr{%
\renewcommand\rmdefault{wncyr}%
\renewcommand\sfdefault{wncyss}%
\renewcommand\encodingdefault{OT2}%
\normalfont
\selectfont}
\DeclareTextFontCommand{\textcyr}{\cyr}
\theoremstyle{plain}
\newtheorem{theorem}{Theorem}
\newtheorem{proposition}{Proposition}

\theoremstyle{definition}
\newtheorem{definition}{Definition}

\newtheorem{assumption}{Assumption}

\theoremstyle{remark}

\begin{document}
\title{The Valuation of Clean Spread Options:  Linking Electricity, Emissions and Fuels}

\author{Ren\'{e} Carmona}
\address{Bendheim Center for Finance\\
Dept. ORFE, University of Princeton\\Princeton NJ 08544, USA}
\email{rcarmona@princeton.edu}

\author{Michael Coulon}
\thanks{Partially supported by NSF - DMS-0739195}
\address{ORFE\\University of Princeton\\Princeton NJ 08544, USA}
\email{mcoulon@princeton.edu}

\author{Daniel Schwarz}
\address{Oxford-Man Institute\\University of Oxford\\Oxford, UK}
\email{schwarz@maths.ox.ac.uk}

\maketitle

\begin{abstract}
The purpose of the paper is to present a new  pricing method for clean spread options, and to illustrate its main features on a set
of numerical examples produced by a dedicated computer code. The novelty of the approach is embedded in the use of structural
models as opposed to reduced-form models which fail to capture properly the fundamental dependencies between the economic factors entering the production process. 
\end{abstract}

\section{Introduction}
Spread options are most often used in the commodity and energy markets to encapsulate the profitability of a production process
by comparing the price of a refined product to the costs of production including, but not limited to, the prices of the inputs to
the production process. When the output commodity is electric power, such spread options are called \emph{spark spreads} when the electricity is produced from natural gas, and \emph{dark spreads} when the electricity is produced from coal. Both processes are the sources of GreenHouse Gas (GHG) emissions, in higher quantities for the latter than the former. 
In this paper we concentrate on the production of electricity and CO${}_2$ emissions and the resulting dependence structure between prices.

Market mechanisms aimed at controlling CO${}_2$ emissions have been implemented throughout the world, and whether they are mandatory or voluntary, cap-and-trade schemes have helped to put a price on carbon in the US and in Europe.
In the academic literature, equilibrium models have been used to show what practitioners have known all along, namely that the price put on CO${}_2$ by the regulation should be included in the costs of production to set the price of electricity. (cf. \cite{CarmonaFehrHinzPorchet}) Strings of spark spread options (options on the spread between the price of 1MWh of electricity and the cost of the amount of natural gas needed to produce such a MWh) with maturities covering a given period are most frequently used to value the optionality of a gas power plant which can be run when it is profitable to do so (namely when the price of electricity is greater than the cost of producing it), and shut down otherwise. In a nutshell, if an economic agent takes control on day $t$, of a gas power plant for a period $[T_1,T_2]$, then for every day
$\tau\in[T_1,T_2]$ of this period, he or she can decide to run the power plant when $P_\tau>h_gS^g_\tau+K$ and book a profit $ P_\tau-h_gS^g_\tau-K$ for each unit of power produced, and shut the plant down if $P_\tau\le h_gS^g_\tau+K$. Moreover, ignoring constraints such as ramp-up rates and start-up costs, this scheduling is automatically induced when generators bid at the level of their production costs in the day-ahead auction for power. Here $P_\tau$ denotes the price at which one unit (1 MWh) of power can be sold on day $\tau$, $S^g_\tau$ the price of one unit of natural gas (typically one MMBtu), $h_g$ the efficiency or heat rate of the plant (i.e. the number of units of natural gas needed to produce one unit of electricity) and $K$ the daily (fixed) costs of operations and maintenance
of the plant. So in this somewhat oversimplified analysis of the optionality of the plant, the value at time $t$ of the control of the plant operation on day $\tau$ can be expressed as 
$e^{-r(\tau-t)}\E[(P_\tau - h_gS^g_\tau-K)^+|\mathcal{F}_t]$ where as usual, the exponent ${}^+$ stands for the positive part, i.e. $x^+=x$ when $x\ge 0$ and $x^+=0$ otherwise, $r$ for the constant interest
rate used as discount factor to compute present values of future cash flows, and $\mathcal{F}_t$ denotes the information available on day $t$ when the conditional expectation is actually computed. So the operational control (for example as afforded by a tolling contract) of the plant over the period $[T_1,T_2]$ could be valued on day $t$ as
$$ 
V^{PP}_t=\sum_{\tau=T_1}^{T_2}e^{-r(\tau-t)}\E[(P_\tau - h_gS^g_\tau-K)^+|\mathcal{F}_t].
$$
This rather simplistic way of valuing a power generation asset in the spirit of the theory of real options, had far-reaching implications in the developments of the energy markets, and is the main reason why spread options are of the utmost importance. However, such a valuation procedure is flawed in the presence of emission regulation as the costs of production have to include the costs specific to the regulation. To be more specific, the day-$\tau$ potential profit $(P_\tau - h_gS^g_\tau-K)^+$ of the spark spread has to be modified to $(P_\tau - h_gS^g_\tau- e_gA_\tau-K)^+$ in order to accommodate the cost of the regulation. Here $A_\tau$ is the price of one allowance certificate worth one ton of CO${}_2$ equivalent, and $e_g$ is the emission coefficient of the plant, namely the number of tons of CO${}_2$ emitted by the plant during the production of one unit of electricity. Such a spread is often called a \emph{clean spread} to emphasize the fact the externality is being paid for, and the real option approach to power plant valuation leads to the following \emph{clean price}
$$
V^{CPP}_t=\sum_{\tau=T_1}^{T_2}e^{-r(\tau-t)}\E[(P_\tau - h_gS^g_\tau- e_gA_\tau-K)^+|\mathcal{F}_t]
$$
for the control of the plant over the period $[T_1,T_2]$ in the presence of the regulation.

In order to price such cross-commodity derivatives, a joint model is clearly required for fuel prices, electricity prices and carbon allowance prices.  Various studies have analyzed the strong links between these price series (cf. \cite{deJongSchneider,pKoenig2011}).  Many reduced-form price models have been proposed for electricity (cf. \cite{fBenth2008,aEydeland2003} for examples) with a focus on capturing its stylized features such as seasonality, high volatility, spikes, mean-reversion and fuel price correlation.  On the other hand, many authors have argued that these same features are better captured via a structural approach, modelling the dynamics of underlying factors such as demand (load), capacity and fuel prices (early examples include \cite{mBarlow2002,aCartea2008,cPirrong2008,mCoulon2009a}).

Similarly, for carbon emission allowances, exogenously specified processes that model prices directly have been proposed by some (cf. \cite{CarmonaHinz}). Others have instead treated the emissions process as the exogenously specified underlying factor; in this case the allowance certificate becomes a derivative on cumulative emissions (cf. \cite{jSeifert2008,ChesneyTaschini}). However, these models do not take into account the important feedback from the allowance price to the rate at which emissions are produced in the electricity sector --- a feature, which is crucial for the justification of any implementation of a cap-and-trade scheme. In a discrete-time framework this feedback mechanism has been addressed, for example in \cite{mCoulon2009,CarmonaFehrHinzPorchet}. In continuous-time the problem has been treated in \cite{CDET} and \cite{sHowison2012}, whereby the former models the accumulation of emissions as a function of an exogenously specified electricity price process, while the latter uses the bid-stack mechanism to infer the emissions rate.

The literature on spread options is extensive. In industry, Margrabe's classical spread option formula (cf. \cite{wMargrabe1978}) is still widely used, and has been extended by various authors (see \cite{CarmonaDurrleman} for an overview) including to the three commodity case, as required for the pricing of clean spreads (cf. \cite{eAlos2011}). \cite{CarmonaSun} analyse the pricing of two-asset spread options in a multiscale stochastic volatility model. For electricity markets, pricing formulae for dirty spreads based on structural models have been proposed in \cite{rCarmona2012}, in which a closed-form formula is derived in the case of $K=0$, and in \cite{rAid2012}, in which semi-closed form formulae are derived for $K\neq0$ at the expense of a fixed merit order.

The original contributions of the paper are twofold. First, we express the value of clean spread options in a formulation where demand for power and fuel prices are the only factors whose stochastic dynamics are given exogenously, and where the prices of power and emission allowances are derived from a bid-stack based structural model and a forward backward stochastic differential system respectively. The second contribution is the development of a numerical code for the computation of the solution of 
the pricing problem. First we solve a 4+1 dimensional semilinear partial differential equation to compute the price of an emission allowance, and then we use Monte Carlo techniques to compute the price of the spread option. These computational tools are used to produce the numerical results of case studies presented in \textsection \ref{str:num_analysis} of the paper for the purpose of illustrating the impact of a carbon regulation on the price of spread options. In this section we first compare the price of spark and dark spread options in two different markets, one with no emissions regulation in place and the other governed by an increasingly strict cap-and-trade system. Second, we analyze the impact that different merit order scenarios have on the option prices. Third, we demonstrate the difference between the structural and the reduced-form approach by comparing the option prices produced by our model with those produced by two key candidate reduced-form models. Fourth and last, we contrast two competing policy instruments: cap-and-trade, represented by the model we propose, and a fixed carbon tax.

\section{The Bid Stack: Price Setting in Electricity Markets}
In order to capture the dependency of electricity price on production costs and fundamental factors in a realistic manner, we use a structural model in the spirit of those reviewed in the recent survey of \cite{CarmonaCoulon}. The premises of structural models for electricity prices depend upon an explicit construction of the supply curve. Since electricity is sold at its marginal cost, the electricity spot price is given by evaluation of the supply function for the appropriate values of factors used to describe the costs of production in the model.

In practice, electricity producers submit day-ahead bids to a central market operator, whose task it is to allocate the production of electricity amongst them. Typically, firms' bids have the form of price-quantity pairs, with each pair comprising the amount of electricity the firm is willing to produce, and the price at which the firm is willing to sell this quantity. Given the large number of generators in most markets, it is common in structural models to approximate the resulting step function of market bids by a continuous increasing curve. Firms' bid levels are determined by their costs of production. An important feature of our model, distinguishing it from most of the commonly used structural models is to include, as part of the production costs, the costs incurred because of the existence of an emissions regulation. 

\vskip 2pt
We assume that, when deciding which firms to call upon to produce electricity, the market operator adheres to the merit order, a rule by which cheaper production units are called upon before more expensive ones. For simplicity, operational and transmission constraints are not considered.

\begin{assumption}
The market operator arranges bids according to the \textbf{merit order}, in increasing order of production costs.
\end{assumption}

The map resulting from ordering market supply in increasing order of electricity costs of production is what is called the bid stack. As it is one of the important building blocks of our model, we define it in a formal way for later convenience.

\begin{definition}
The \textbf{bid stack} is given by a measurable function
\begin{equation*}
 b:\left[0,\bar{x}\right]\times\R\times\R^n \ni (x,a,s) \hookrightarrow b(x,a,s) \in\R,
\end{equation*}
with the property that for each fixed $(a,s)\in\R\times\R^n$, the function $[0,\bar{x}]\ni x\hookrightarrow b(x,a,s)$ is strictly increasing.
\end{definition}
In this definition,  $\bar{x}\in \R_{++}$ represents the \textit{market capacity} (measured in MWh) and the variable $x$ the \textit{supply of electricity}. The integer $n \in \mathbb{N}\setminus \{0\}$ gives the number of economic factors (typically the prices  in \euro\  of the fuels used in the production of electricity), and $s\in\R^n$ the numeric values of these factors.
Here and throughout the rest of the paper the \textit{cost of carbon emissions} (measured in \euro\ per metric ton of CO$_2$) is denoted by $a$. So for a given allowance price, say $a$, and fuel prices, say $s$, the market is able to supply $x$ units of electricity at price level $b=b(x,a,s)$ (measured in \euro\ per MWh).  In other words, $b(x,a,s)$ represents the bid level of the marginal production unit in the event that demand equals $x$.

The choice of a function $b$ which captures the subtle dependence of the electricity price upon the level of supply and the production costs, is far from trivial, and different approaches have been considered in the literature, as reviewed recently by \cite{CarmonaCoulon}.  In \textsection \ref{str:num_bid_stack} we extend the model proposed in \cite{rCarmona2012} to include the cost of carbon as part of the variable costs driving bid levels.  

\section{Risk-Neutral Pricing of Allowance Certificates}
As the inclusion of the cost of emission regulation in the valuation of spread options is the main thrust of the paper, we explain how emission allowances are priced in our model. The model we introduce is close to \cite{sHowison2012}. However we extend the results found therein to allow the equilibrium bids of generators to be stochastic and driven by fuel prices, a generalization that is vital for our purpose.

We suppose that carbon emissions in the economy are subject to cap-and-trade regulation structured as follows: at the end of the compliance period, each registered firm needs to offset its cumulative emissions with emission allowances or incur a penalty for each excess ton of CO${}_2$ not covered by a redeemed allowance certificate. Initially, firms acquire allowance certificates through free allocation, e.g. through National Allocation Plans (NAP) like in the initial phase of the European Union (EU) Emissions Trading Scheme (ETS), or by purchasing them at auctions like in the Regional Greenhouse Gas Initiative (RGGI) in the North East of the US.
Allowances change hands throughout the compliance period. Typically, a firm which thinks that its initial endowment will not suffice to cover its emissions will buy allowances, while firms expecting a surplus will sell them. Adding to these \emph{naturals}, speculators enter the market providing liquidity. Allowances are typically traded in the form of forward contracts and options. In this paper, we denote by $A_t$ the spot price of an allowance certificate maturing at the end of the compliance period. Because their cost of carry is negligible, we treat them as financial products  liquidly traded in a market without frictions, and in which long and short positions can be taken. 

In a competitive equilibrium, the level of cumulative emissions relative to the cap (i.e. the number of allowance certificates issued by the regulation authority) determines whether --- at the end of the compliance period --- firms will be subjected to a penalty payment and create a demand for allowance certificates. See \cite{CarmonaFehrHinzPorchet} for details. For this reason, allowance certificates should be regarded as derivatives on the emissions accumulated throughout the trading period. This type of option written on a non-tradable underlying interest is rather frequent in the energy markets: temperature options are a case in point.

\subsection{The Market Emissions Rate}
As evidenced by the above discussion, the rate at which CO${}_2$ is emitted in the atmosphere as a result of electricity production has to be another important building block of our model. Clearly at any given time,  this rate is a function of the amount of electricity produced  and because of their impact on the merit order, the variable costs of production, including fuel prices, and notably, the carbon allowance price itself.

\renewcommand{\theenumi}{D\thedefinition.\arabic{enumi}}
\begin{definition}\label{def:market_ems}
 The \textbf{market emissions rate} is given by a bounded function
\begin{equation*}
 \mu_e:\left[0,\bar{x}\right]\times\R\times\R^n \ni (x,a,s) \hookrightarrow \mu_e(x,a,s) \in\R_{+},
\end{equation*}
which is Lipschitz continuous in its three variables, strictly increasing in $x$ when $a$ and $s$ are held fixed, and strictly decreasing in $a$ when $x$ and $s$ are fixed.  

\end{definition}

With the definition above, for a given level of electricity supply and for given allowance and fuel prices, $\mu_e=\mu_e(x,a,s)$ represents the rate at which the market emits, measured in tons of CO$_2$ per hour. Cumulative emissions are then computed by integrating the market emissions rate over time.
The monotonicity property in $x$ makes sense since any increase in supply can only increase the emissions rate. Similarly, as the cost of carbon increases the variable costs (and hence the bids) of pollution intensive generators increase by more than those of environmentally friendlier ones. Dirtier technologies become relatively more expensive and are likely to be scheduled further down in the merit order. As a result cleaner technologies are brought online earlier, hence the monotonicity in $a$.

In \textsection \ref{str:num_ems_stack} we propose a specific functional form for $\mu_e$ consistent with the bid stack model introduced in \textsection\ref{str:num_bid_stack}.  

\subsection{The Pricing Problem}
We shall use the following notation. For a fixed time horizon $T \in \R_+$, let $(W^0_t,W_t)_{t\in[0,T]}$ be a  $(n+1)$-dimensional standard Wiener process on a probability space $(\Omega, \mathcal{F},\mathbb{P})$, $\mathcal{F}^0:=(\mathcal{F}^0_t)$ the filtration generated by $W^0$, $\mathcal{F}^W:=(\mathcal{F}^W_t)$  the filtration generated by $W$, and $\mathcal{F}:=\mathcal{F}^0\vee\mathcal{F}^W$ the market filtration. All relationships between random variables are to be understood in the almost surely sense.

Consumers' \textit{demand for electricity} is given by an $\mathcal{F}^0_t$-adapted stochastic process $(D_t)$. In response to this demand producers supply electricity, and we assume that demand and supply are always in equilibrium, so that at any time $t\in[0,T]$ an amount $D_t$ of electricity is supplied. The \textit{prices of fuels} are observed $\mathcal{F}^W_t$-adapted stochastic processes $(S_t)_{t\in[0,T]}$, where $S_t:=(S^1_t,\ldots,S^n_t)$. If the price of an \textit{allowance certificate} at time $t$, say $A_t$, becomes available,  as we will see in \textsection \ref{str:allowance_FBSDE}, $(A_t)_{T\in[0,T]}$ will be constructed as a $\mathcal{F}_t$-adapted stochastic process solving a Forward Backward Stochastic Differential Equation (FBSDE). The rate of emission $\mu_e(D_t,A_t,S_t)$ can then be evaluated and the cumulative emissions computed by integration over time, resulting in a $\mathcal{F}_t$-adapted process $(E_t)_{t\in[0,T]}$.

In order to avoid the difficulties of estimating the market price of risk (see for example \cite{aEydeland2003} for a discussion of some possible ways to approach this thorny issue), we choose to specify the dynamics of the processes  $(D_t)_{t\in[0,T]}$ and  $(S_t)_{t\in[0,T]}$ under a risk neutral
measure $\mathbb{Q} \sim \mathbb{P}$ chosen by the market for pricing purposes.

\subsection{An FBSDE for the Allowance Price}\label{str:allowance_FBSDE}
We assume that at time $t=0$, demand for electricity is known. Thereafter, it evolves according to an It\^{o} diffusion. Specifically, for $t\in[0,T]$, demand for electricity $D_t$ is the unique strong solution of a stochastic differential equation of the form
\begin{equation}\label{eq:demand_process}
 \Id D_t = \mu_d(t,D_t) \Id t + \sigma_d(D_t) \Id \tilde{W}^0_t, \qquad D_0=d_0 \in (0,\bar{x}),
\end{equation}
where $(\tilde{W}_t)$ is an $\mathcal{F}_t$-adapted $\mathbb{Q}$-Brownian motion. The time dependence of the drift allows us to capture the seasonality observed in electricity demand.

Similarly to demand, the prices of the fuels used in the production processes satisfy a system of stochastic differential equations written in a vector form as follows:
\begin{equation}\label{eq:fuel_price_process}
 \Id S_t = \mu_s(S_t) \Id t + \sigma_s(S_t) \Id \tilde{W}_t, \qquad S_0=s_0 \in \R^n,\; t\in[0,T].
\end{equation}
Cumulative emissions are measured from the beginning of the compliance period when time $t=0$, so that $E_0=0$. Subsequently, they are determined by integrating over the market emissions rate $\mu_e$ introduced in Definition \ref{def:market_ems}. So assuming that the price $A_t$ of an allowance certificate is knwon, the cumulative emissions process is represented by a bounded variation process; i.e. for $t\in[0,T]$,
\begin{equation}\label{eq:cumulative_emissions_process}
 \Id E_t = \mu_e(D_t,A_t,S_t) \Id t, \quad E_0=0.
\end{equation}
Note that with this definition the process $(E_t)$ is non-decreasing, which makes intuitive sense considering that it represents a cumulative quantity.

To complete the formulation of the pricing model, it remains to characterize the allowance certificate price process $(A_t)_{t\in[0,T]}$. If our model is to apply to a one compliance period scheme, in a competitive equilibrium, at the end of the compliance period $t=T$, its value is given by a deterministic function of the cumulative emissions:
\begin{equation}\label{eq:allowance_terminal_cnd}
A_T = \phi(E_T),
\end{equation}
where $\phi:\R \hookrightarrow \R$ is bounded, measurable and non-decreasing. Usually $\phi(\cdot):=\pi \mathbb{I}_{[\Gamma,\infty)}(\cdot)$, where $\pi\in\R_+$ denotes the penalty paid in the event of non-compliance and $\Gamma\in \R_+$ the cap chosen by the regulator as the aggregate allocation of certificates. See \cite{CarmonaFehrHinzPorchet} for details.
Since the discounted allowance price is a martingale under $\mathbb{Q}$, it is equal to the conditional expectation of its terminal value, i.e.
\begin{equation}\label{eq:allowance_risk_neutral_expectation}
 A_t=\exp\left(r(T-t)\right)\E^\mathbb{Q}\left[\left.\phi(E_T)\right|\mathcal{F}_t\right], \quad \text{for } t\in[0,T],
\end{equation}
which implies in particular that the allowance price process $(A_t)$ is bounded.
Since the filtration $(\mathcal{F}_t)$ is being generated by the Wiener processes, it is a consequence of the Martingale Representation Theorem (cf. \cite{iKaratzas1999}) that the allowance price can be represented as an It\^{o} integral with respect to the Brownian motion $(\tilde{W}^0_t,\tilde{W}_t)$. It follows that
\begin{equation}\label{eq:allowance_process}
 \Id A_t = r A_t \Id t +Z^0_t \Id \tilde{W}^0_t + Z_t \cdot \Id \tilde{W}_t, \quad \text{for } t\in[0,T]
\end{equation}
for some $\mathcal{F}_t$-adapted square integrable process $(Z^0_t,Z_t)$.

Combining equations \eqref{eq:demand_process}, \eqref{eq:fuel_price_process}, \eqref{eq:cumulative_emissions_process}, \eqref{eq:allowance_terminal_cnd} and \eqref{eq:allowance_process}, the pricing problem can be reformulated as the solution of the FBSDE
\begin{equation}\label{eq:allowance_FBSDE}
\left\{
 \begin{aligned}
 \Id D_t &= \mu_d(t,D_t) \Id t + \sigma_d(D_t) \Id \tilde{W}^0_t,		&& D_0=d_0\in (0,\bar{x}),\\
 \Id S_t &= \mu_s(S_t) \Id t + \sigma_s(S_t) \Id \tilde{W}_t,			&& S_0=s_0\in \R^n,\\
 \Id E_t &= \mu_e(D_t,A_t,S_t) \Id t,						&& E_0=0,\\
 \Id A_t &=r A_t \Id t + Z^0_t \Id \tilde{W}^0_t + Z_t\cdot \Id \tilde{W}_t,	&& A_T = \phi(E_T).
 \end{aligned}
\right.
\end{equation}
Notice that the first two equations are standard stochastic differential equations (in the forward direction of time) which do not depend upon the cumulative emissions and the allowance price. We will choose their coefficients so that existence and uniqueness of solutions hold. To be more specific we make the following assumptions on the coefficients of \eqref{eq:allowance_FBSDE}:
\renewcommand{\theenumi}{A\theassumption.\arabic{enumi}}
\begin{assumption}
\label{as:DSexistence}
The functions $\mu_d : [0,T]\times [0,\bar{x}] \hookrightarrow \R$, $\sigma_d : [0,\bar{x}] \hookrightarrow \R$, $\mu_s : \R^n \hookrightarrow \R^n$, $\sigma_s : \R^n \hookrightarrow \R^n\times \R^n$ are such that the first two equations in \eqref{eq:allowance_FBSDE} have a unique strong solution.
\end{assumption}

\subsection{Existence of a Solution to the Allowance Pricing Problem}
\begin{theorem}
\label{th:existence}
We assume that Assumption \ref{as:DSexistence} holds and that $\mu_e$ is Lipshitz with respect to the variable $a$ uniformly in $x$ and $s$, and that $\mu_e(x,0,s)$ is uniformly bounded in $x$ and $s$. Then if $\phi$ is bounded and Lipschitz, the FBSDE \eqref{eq:allowance_FBSDE} has a unique square integrable solution.
\end{theorem}
\begin{proof}
Assumption \ref{as:DSexistence} being satisfied, and the first two equations of \eqref{eq:allowance_FBSDE} being decoupled from the remaining ones, there exist adapted processes $(D_t)$ and $(S_t)$ with values in $[0,\bar{x}]$ and
$\R^n$ respectively, unique strong solutions of the first two equations of \eqref{eq:allowance_FBSDE}. Once these two processes are
constructed, we can plug their values into the last two equations of \eqref{eq:allowance_FBSDE}, and treat the resulting equations as an FBSDE with random coefficients. Existence and uniqueness hold because of Theorem 7.1 of \cite{jMa2011} \footnote[1]{We would like to thank Francois Delarue for suggesting this strategy and the use of \cite{jMa2011} in the present set-up.}. Strictly speaking this result is only proved for one-dimensional processes. In the present situation, while $E_t$ and $A_t$ are indeed one-dimensional, the Wiener process is $(n+1)$-dimensional and we cannot apply directly Theorem 7.1 of \cite{jMa2011}. However, a close look at the proof of this result shows that what is really needed is to prove the well-posedness of the \emph{characteristic BSDE}, and the boundedness of its solution and the solutions of the \emph{dominating Ordinary Differential Equations} (ODE). In the present situation, these equations are rather simple due to the fact that $E_t$ has bounded variation, and as a consequence, its volatility vanishes. The two dominating ODEs can be solved explicitly and one can check that the solutions are bounded by inspection. Moreover, the function $\phi$ for the terminal condition being uniformly Lipschitz, the characteristic BSDE is one-dimensional, though driven by a multi-dimensional Brownian motion, its terminal condition is bounded, and Kobylanski's comparison results (see the original contribution \cite{mKobylanski2000}) can be used to conclude the proof.
\end{proof}

The above existence result is proven for a terminal condition given by a smooth function $\phi$. As already mentioned earlier, single compliance period equilibrium models most often require that the function $\phi$ take two values, and the terminal condition $\phi(E_T)$ equal the penalty when the regulatory cap is exceeded, i.e. when $E_T>\Gamma$, and zero when $E_T< \Gamma$. \cite{CarmonaDelarue} proved that a weaker form of existence and uniqueness of a solution to the FBSDE still holds when $\phi$ is discontinuous (in particular when $\phi$ is an indicator function). Given that the \emph{decoupling field} constructed in \cite{jMa2011} is uniformly Lipschitz, we conjecture that a proof in the spirit of the one given in \cite{CarmonaDelarue} should work here and provide this weaker form of existence and uniqueness. However, Carmona and Delarue also proved that under a strict monotonicity assumption on $\mu_e$ (which should hold in our case for intuitive reasons), the aggregate emissions were equal with positive probability to the cap at the end of the compliance period, and the terminal condition could not be prescribed for all the scenarios. We suspect that in the present situation, the cumulative emissions equal the cap (i.e. $E_T=\Gamma$) for a set of scenarios of positive probability, and the terminal price of an allowance $A_T$ cannot be prescribed in advance on this set of scenarios.

\section{Valuing Clean Spread Options}
In this section we consider the problem of spread option pricing as described in the introduction. Whether the goal is to value a physical asset or risk manage financial positions, one needs to compute the price
of a European call option on the difference between the price of electricity and the costs of production for a particular power plant. The costs that we take into account are the fixed operation and maintenance costs, the cost of the fuel needed to generate one MWh of electricity and the cost of the ensuing emissions. Letting the $\mathcal{F}_t$-adapted process $(P_t)$ denote the spot price of electricity, and recasting the informal discussion of the introduction with the notation we chose to allow for several input fuels, a clean spread option with maturity $\tau\in[0,T]$ is characterized by the payoff 
\begin{equation*}
 \left( P_\tau - h_v S^v_\tau - e_v A_\tau-K \right)^+,
\end{equation*}
where $K$ represents the value of the fixed operation and maintenance costs, $h_v\in \R_{++}$ and $e_v\in\R_{++}$ denote the specific heat and emissions rates of the power plant under consideration, and $S^v\in\{S^1,\ldots,S^n\}$ is the price at time $\tau$ of the fuel used in the production of electricity. In the special case when $S^v$ is the price of coal (gas) the option is known as a \textit{clean dark (spark) spread} option.

Since we are pricing by expectation, the value $V^v_t$ of the clean spread is given by the conditional expectation under the pricing measure of the discounted payoff; i.e.
\begin{equation*}
 V^v_t = \exp(-r(\tau-t))\E^\mathbb{Q}\left[\left(P_\tau - h_v S^v_\tau - e_v A_\tau-K\right)^+|\mathcal{F}_t\right], \quad \text{for } t\in[0,\tau].
\end{equation*}

\section{A Concrete Two-Fuel Model}
We now turn to the special case of two fuels, coal and gas.

\subsection{The Bid Stack}\label{str:num_bid_stack}
Our bid stack model is a slight variation on the one we proposed in \cite{rCarmona2012}. Here we extend to include the cost of emissions as part of the variable costs driving firms' bids.

We assume that the coal and gas generators have aggregate capacities $\bar{x}_c$ and  $\bar{x}_g$ respectively, so that the market capacity is $\bar{x} = \bar{x}_c+\bar{x}_g$, and their bid levels are given by linear functions of the allowance price and the price of the fuel used for the generation of electricity. We denote these bid functions by $b_c$ and $b_g$ respectively.  The coefficients appearing in these linear functions correspond to the marginal emissions rate (measured in ton equivalent of CO$_2$ per MWh) and the heat rate (measured in MMBtu per MWh) of the technology in question. Specifically, for $i\in\{c,g\}$, we assume that
\begin{equation}\label{eq:fuel_bid_curves}
 b_i(x,a,s):=e_i(x)a+h_i(x)s, \quad \text{for } (x,a,s) \in [0,\bar{x}_i]\times \R \times \R,
\end{equation}
where the \textit{marginal emissions rate} $e_i$ and the \textit{heat rate} $h_i$ are given by
\begin{equation*}
\begin{aligned}
 e_i(x)&:=\hat{e}_i\exp\left(m_ix\right)\\
 h_i(x)&:=\hat{h}_i\exp\left(m_ix\right)
\end{aligned}
,\quad \text{for } x\in [0,\bar{x}_i].
\end{equation*}
Here $\hat{e}_i$, $\hat{h}_i$, $m_i$ are strictly positive constants. We allow the marginal emissions rate and the heat rate of each technology to vary to reflect differences in efficiencies within the fleet of coal and gas generators.  Since less efficient plants with higher heat rates have correspondingly higher emissions rates, it is a reasonable approximation to assume that for each technology the ratio $h_i/e_i$ is fixed.

\begin{proposition}\label{prop:two_fuel_market_bs}
With $b_c$ and $b_g$ as above and $I=\{c,g\}$, the market bid stack $b$ is given by 
\begin{equation*}
 b(x,a,s) = \left\{\begin{array}{ll}
 \left(\hat{e}_i a+\hat{h}_i s_i\right)\exp\left(m_i x\right), & \text{ if  } b_i(x,a,s_i)\leq b_j(0,a,s_j) \text{  for } i,j\in I, i\ne j, \\
 \left(\hat{e}_i a+\hat{h}_i s_i\right)\exp\left(m_i (x-\bar x_j)\right), & \text{ if  } b_i(x-\bar x_j,a,s_i)> b_j(0,a,s_j) \text{  for } i,j\in I, i\ne j, \\
 \prod_{i\in I}\left(\hat{e}_i a+\hat{h}_i s_i\right)^{\beta_i}\exp\left(\gamma x\right), & \text{  otherwise} \end{array}\right.
 \end{equation*}
for $(x,a,s) \in [0,\bar{x}]\times \R \times \R^2$, where $\beta_i=\frac{m_{M\setminus \{i\}}}{m_c+m_g}$ and $\gamma=\frac{m_cm_g}{m_c+m_g}$.
\end{proposition}

\begin{proof}
The proof is a straightforward extension of Corollary 1 in \cite{rCarmona2012}.
\end{proof}

\subsection{The Emissions Stack}\label{str:num_ems_stack}
In order to determine the rate at which the market emits we need to know which generators are supplying electricity at any time. By the merit order assumption the market operator calls upon firms in increasing order of their bid levels. Therefore, given electricity, allowance and fuel prices $(p,a,s)\in \R\times\R\times\R^2$, for $i\in\{c,g\}$, the set of active generators of fuel type $i$ is given by $\left\{ x \in [0,\bar{x}_i] : b_i(x,a,s)\leq p\right\}$.
\begin{proposition}
Assuming that the market bid stack is of the form specified in Proposition \ref{prop:two_fuel_market_bs}, the market emissions rate $\mu_e$ is given by
\begin{equation}\label{eq:muE}
\mu_e(x,a,s):=\frac{\hat e_g}{m_g}\left(\exp\left(m_g \hat b_g^{-1}(b(x,a,s),a,s_g)\right)-1\right) + \frac{\hat e_c}{m_c}\left(\exp\left(m_c \hat b_c^{-1}(b(x,a,s),a,s_c)\right)-1\right)
\end{equation}
for $(x,a,s)\in [0,\bar{x}] \times \R \times \R^2$, where for $i\in\{c,g\}$ we define 
\begin{equation*}
 \hat b_i^{-1}(p,a,s_i):=0 \vee \left(\bar x_i \wedge \frac1{m_i}\log\left(\frac{p}{\hat e_i a + \hat h_i s_i}\right)\right),
 \end{equation*}
 for $(p,a,s)\in \R \times \R \times \R^2$.
\end{proposition}

\begin{proof}
The market emissions rate follows from integrating the marginal emissions rate $e_i$ for each technology over the corresponding set of active generators and then summing the two. Given the monotonicity of $b_i$ in $x$ and its range $[0,\bar x_i]$, the function $\hat b_i^{-1}$ describes the quantity of electricity supplied by fuel $i\in\{c,g\}$, and hence the required upper limit of integration.
\end{proof}

\subsection{Specifying the Exogenous Stochastic Factors}\label{str:stoch_proc}

\subsubsection*{The Demand Process}\label{str:dem_proc}
We posit that under $\mathbb{Q}$, the process $(D_t)$ satisfies  for $t\in[0,T]$ the stochastic differential equation
\begin{equation*}
 \Id D_t = -\eta\left(D_t-\bar{D}(t)\right) \Id t + \sqrt{2\eta \hat{\sigma} D_t\left(\bar{x} - D_t\right)} \Id \tilde{W}_t, \qquad D_0=d_0 \in (0,\bar{x}),
\end{equation*}
where $[0,T] \ni t \hookrightarrow \bar{D}(t)\in (0,\bar{x})$ is a deterministic function giving the level of mean reversion of the demande and $\eta, \hat{\sigma} \in\R_{++}$ are constants. With this definition $(D_t)$ is a Jacobi diffusion process; it has a linear, mean-reverting drift component and degenerates on the boundary. Moreover, subject to $\min(\bar{D}(t),\bar{x}-\bar{D}(t))\geq \bar{x}\hat{\sigma}$, for $t\in [0,T]$, the process remains within the interval $(0,\bar{x})$ at all times (cf. \cite{jForman2008}). To capture the seasonal character of demand, we choose a function $\bar{D}(t)$ of the form:
\begin{equation*}
 \bar{D}(t):=\varphi_0 + \varphi_1 \sin(2\pi\vartheta t),
\end{equation*}
where the coefficients will be chosen below.

\subsubsection*{The Fuel Price Processes}
 We assume that the prices of coal $(S^c_t)$ and gas $(S^g_t)$ follow correlated exponential OU processes under the measure $\mathbb{Q}$; i.e., for $i \in \{c,g\}$ and $t\in[0,T]$,
\begin{equation*}
\Id S^i_t = -\eta_i \left(\log S^i_t-\bar{s}_i - \frac{\hat{\sigma}_i^2}{2 \eta_i}\right)S^i_t \Id t + \hat{\sigma}_i S^i_t \Id \tilde{W}^i_t, \quad S^i_0=s^i_0 \in \R_{++},
\end{equation*}
where $\Id\left<W^c,W^g\right>_t = \rho \Id t$.

\section{Numerical Analysis}\label{str:num_analysis}
We now turn to the detailed analysis of the model we propose. For this purpose we consider a number of case studies in \textsection \ref{str:CvsD} to \textsection \ref{str:CATvsCT}. To produce the following results we used the numerical schemes explained in Appendix \ref{ap:num_soln_FBSDE} and Appendix \ref{ap:num_soln_spread}.

\subsection{Choice of Parameters}
The tables in this section summarise the parameters used for the numerical analysis of our model that follows below. We refer to the parameters specified in Tables \ref{tab:numpar_stack} - \ref{tab:numpar_spread} as the `base case' and indicate whenever we depart from this choice. Note that our parameter choices do not correspond to a particular electricity market, but that all values are within a realistic realm.\\

Table \ref{tab:numpar_stack} summarises the parameters specifying the bid curves. We consider a medium sized electricity market served by coal and gas generators and with gas being the dominant technology. For the marginal emission rates, Table \ref{tab:numpar_stack} implies that $e_c\in[0.9,1.64]$ and $e_g\in[0.4,0.69]$ (both measured in tCO$_2$ per MWh), so that all gas plants are `cleaner' than all coal plants. For the heat rates, we observe that $h_c\in[3,5.5]$ and $h_g\in[7,12]$ (both measured in MMBtu per MWh). Using \eqref{eq:muE} now with $D_t=\bar x$, for $0\leq t \leq T$, and the assumption that there are 8760 production hours in the year, we find, denoting the the maximum cumulative emissions by $\bar{e}$, that $\bar{e} = 2.13e+08$. \\

\begin{table}
 \begin{center}
  \begin{tabular}{ccccccccc}
\hline
    $\hat{h}_c$ & $\hat{e}_c$ & $m_c$ & $\bar{x}_c$ & $\hat{h}_g$ & $\hat{e}_g$ & $m_g$ & $\bar{x}_g$ & $\bar{x}$ \\
\hline
     3 & 0.9 & 0.00005 & 12000 & 7 & 0.4 & 0.00003 & 18000 & 30000\\
\hline
  \end{tabular}
 \end{center}
 \caption{Parameters relating to the bid stack and the emissions stack.}
 \label{tab:numpar_stack}
\end{table}
Table \ref{tab:numpar_demand} contains the parameters for the demand process $(D_t)$. We model periodicities on an annual and a weekly time scale and the chosen rate of mean reversion assumes that demand reverts to its (time dependent) mean over the course of one week.
\begin{table}
 \begin{center}
  \begin{tabular}{cccccccc}
    \hline
    $\eta$ & $\varphi_0$ & $\varphi_1$ & $\vartheta$ & $\hat{\sigma}$ & $d_0$ \\
    \hline
     50 & 21000 & 3000 & 1 & 0.1 & 21000\\
    \hline
    \end{tabular}
 \end{center}
 \caption{Parameters relating to the demand process.}
 \label{tab:numpar_demand}
\end{table}

In Table \ref{tab:numpar_fuel} we summarise the parameters that specify the behavior of the prices of coal and gas. Both are chosen to be slowly mean-reverting, at least in comparison to demand. To ease analysis, we assume that all parameters are identical for the two fuels, including mean price levels, both measured in MMBtu.  
 
\begin{table}
 \begin{center}
  \begin{tabular}{ccccccccc}
\hline
    $\eta_c$ & $\bar{s}_c$ & $\hat{\sigma}_c$ & $s^c_0$ & $\eta_g$ & $\bar{s}_g$ & $\hat{\sigma}_g$ & $s^g_0$ & $\rho$\\
\hline
     1.5 & 2 & 0.5 & $\exp(2)$ & 1.5 & 2 & 0.5 & $\exp(2)$ & 0.3\\
\hline
  \end{tabular}
 \end{center}
 \caption{Parameters relating to the fuel price processes.}
 \label{tab:numpar_fuel}
\end{table}

Table \ref{tab:numpar_cat} defines the cap-and-trade scheme that we assume to be in place. The duration of the compliance period $T$ is measured in years and we set the cap at 70\% of the upper bound $\bar{e}$ for the cumulative emissions, in order to incentivise a reduction in emissions.  This choice of parameters results in $A_0$ being approximately equal to $\pi/2$, a value for which there is significant initial overlap between gas and coal bids in the stack.  Furthermore, the parameters imply a bid stack structure such that at mean levels of coal and gas prices, $A_t=0$ pushes all coal bids below gas bids, while for $A_t=\pi$ almost all coal bids are above all gas bids.  
\begin{table}
 \begin{center}
  \begin{tabular}{cccc}
\hline
    $\pi$ & $\Gamma$ & $T$ & $r$ \\
\hline
     100 & 1.4e+08 & 1 & 0.05 \\
\hline
  \end{tabular}
 \end{center}
 \caption{Parameters relating to the cap-and-trade scheme.}
 \label{tab:numpar_cat}
\end{table}

Finally, in Table \ref{tab:numpar_spread} we specify the four spread option contracts used in the base case scenario to represent high and low efficiency coal plants, and high and low efficency gas plants. (Note that low efficiency means dirtier and corresponds to high $h_v$ and $e_v$ and vice versa.)  
\begin{table}
 \begin{center}
  \begin{tabular}{cccccccc}
\hline
\multicolumn{2}{c}{High Eff. Coal} & \multicolumn{2}{c}{Low Eff. Coal} & \multicolumn{2}{c}{High Eff. Gas} & \multicolumn{2}{c}{Low Eff. Gas} \\
\hline
    $h_c$ & $e_c$ & $h_c$ & $e_c$ &$h_g$ & $e_g$ & $h_g$ & $e_g$ \\
\hline
     3.5 & 1.05  & 5.0 &  1.5 &   7.5 & 0.43  & 11.5& 0.66 \\
\hline
  \end{tabular}
 \end{center}
 \caption{Parameters relating to the spread options.}
 \label{tab:numpar_spread}
\end{table}

We now consider a series of case studies to investigate various features of the model's results in turn.  As the model captures many different factors and effects, this allows us to isolate some of the most important implications.  In Case Study I, we investigate the impact on coal and gas plants of different efficiencies of creating an increasingly strict carbon emissions market.  In Case Study II, we assess the impact on these plants of changes in initial fuel prices.  In Case Study III, we compare spread option prices in our model with two simple reduced-form approaches for $A_t$, which allows us to better understand the role of key model features such as bid stack driven abatement.  Finally in Case Study IV, we consider the overall impact of cap-and-trade markets in the electricity sector, by comparing with a well-known alternative, a fixed carbon tax.  

\subsection{Case Study I:  Impact of the Emissions Market}\label{str:CvsD}

The first effect that we are interested in studying in the model is the impact of the cap-and-trade market on clean spread option prices, for increasingly strict levels of the cap $\Gamma$.  At one extreme (when the cap is so generous that $A_t\approx 0$, for all $t\in[0,T]$), results correspond to the case of a market without a cap and trade system, while at the other extreme (when the cap is so strict that $A_t\approx \pi\exp(-r(T-t))$, for all $t\in[0,T]$), there is essentially a very high carbon tax which tends to push most coal generators above gas generators in the stack.  It is intuitively clear that higher carbon prices typically lead to higher spark spread option prices and lower dark spread option prices, thus favouring gas plants over coal plants, but the relationships can be more involved as they vary between low efficiency and high efficiency plants.  

\begin{figure}
  \centering
  \includegraphics[width=\textwidth]{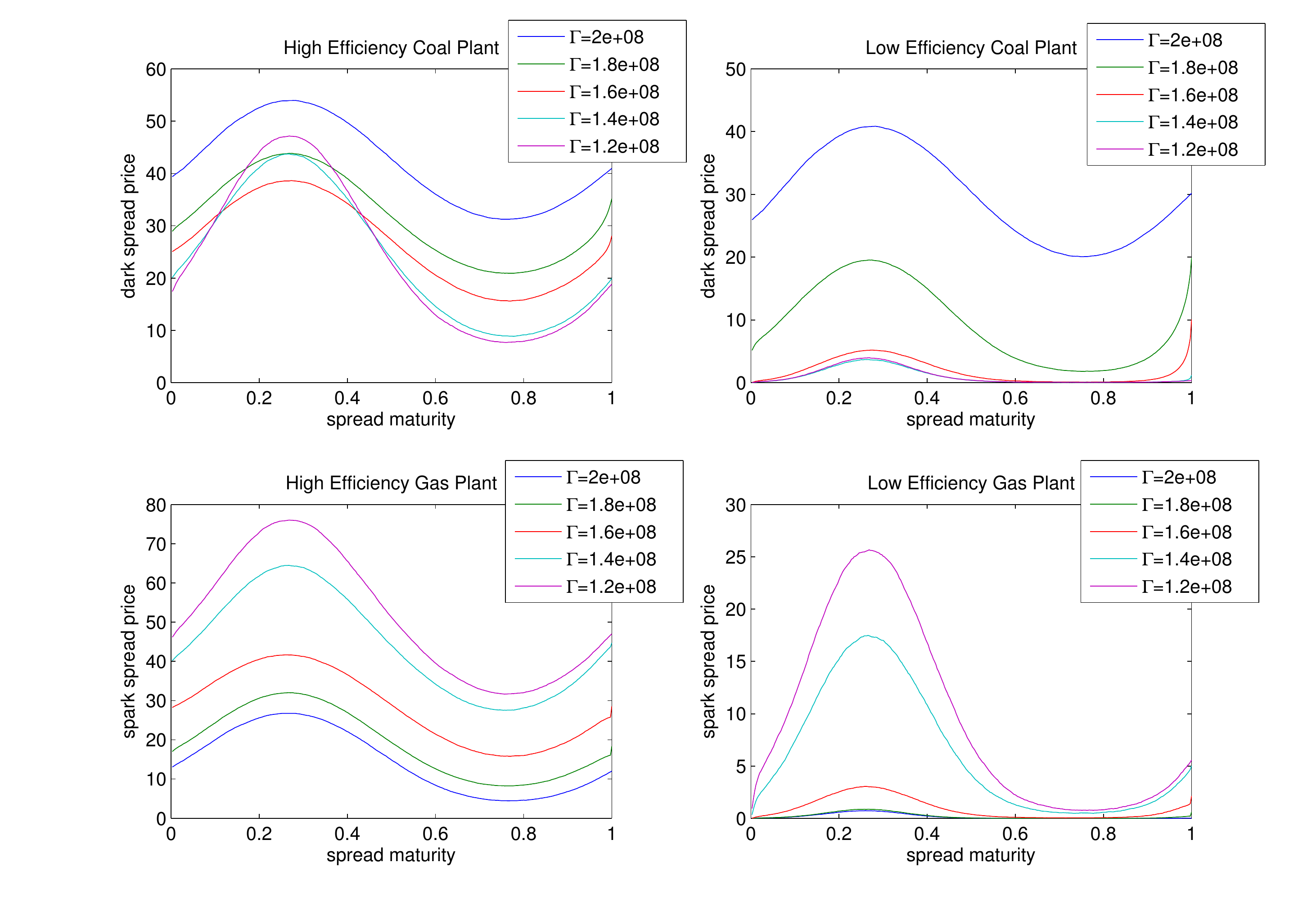}
  \caption{Cap strictness analysis for high efficiency coal (top left), low efficiency coal (top right), high efficiency gas (bottom left) and low efficiency gas (bottom right):  Spark and dark spread option values plotted against maturity, for varying levels of the cap $\Gamma$.  Note that the five equally-spaced cap values from $2e+08$ to $1.2e+08$ tons of CO$_2$ imply initial allowance prices of \$5, \$28, \$52, \$80, and \$94. } 
 \label{fig:casestudy1}
\end{figure}
In Figure \ref{fig:casestudy1}, we compare spread option prices corresponding to different efficiency generators (i.e., to different $h_v,e_v$ in the spread payoff) as a function of maturity $\tau$.  `High' and `low' efficiency plant indicates values of $h_v,e_v$ chosen to be near the lowest and highest respectively in the stack, as given by Table \ref{tab:numpar_spread}.  Within each of the four subplots, the five lines correspond to five different values of the cap $\Gamma$, ranging from very lenient to very strict.  We immediately observe in Figure \ref{fig:casestudy1} the seasonality in spread prices caused by the seasonality in power demand.  This is most striking for the low efficiency cases (high $h_v,e_v$), as such plants would rarely be used in shoulder months, particularly in the case of gas. For low efficiency plants, the relationship with cap level (and corresponding initial allowance price) is as one would expect: a stricter cap greatly increases the value of gas plants and greatly decreases the value of the dirtier coal plants.    This is also true for high efficiency gas plants, although the price difference (in percentage terms) for different $\Gamma$ is less, since these are effectively `in-the-money' options, unlike those discussed above. However, the analysis becomes more complicated for high efficiency coal plants, which tend to be chosen to run in most market conditions, irrespective of emissions markets.  Interestingly, we find that for these options the relationship with $\Gamma$ (and hence $A_0$) can be non-monotonic under certain conditions, particularly for high levels of demand, when the price is set near the very top of the stack.  In such cases a stricter cap provides extra benefit for the cleaner coal plants via higher power prices (typically set by the dirtier coal plants on the margin) which outweighs the disadvantage of coal plants being replaced by gas plants in the merit order.

\subsection{Case Study II:  Impact of Fuel Price Changes}

Notice that in Table \ref{tab:numpar_fuel}, the initial conditions of both gas and coal have been set to be equal to their long term median levels.  We now consider the case that gas price $s^g_0$ is either above or below its long term level, thus inducing a change in the initial merit order.  Given the record low prices of under \$2 recently witnessed in the US natural gas market (due primarily to shale gas discoveries), it is natural to ask how such fuel price variations affect our spread option results.  Note however that since $\eta_c=\eta_g=1.5$ (implying a typical mean reversion time of 8 months), by the end of the trading period, the simulated fuel price distributions will again be centred near their mean reversion levels.  Thus in this case study, we capture a temporary, not permanent, shift in fuel prices.  

\begin{figure}
  \centering
 \includegraphics[width=0.45\textwidth]{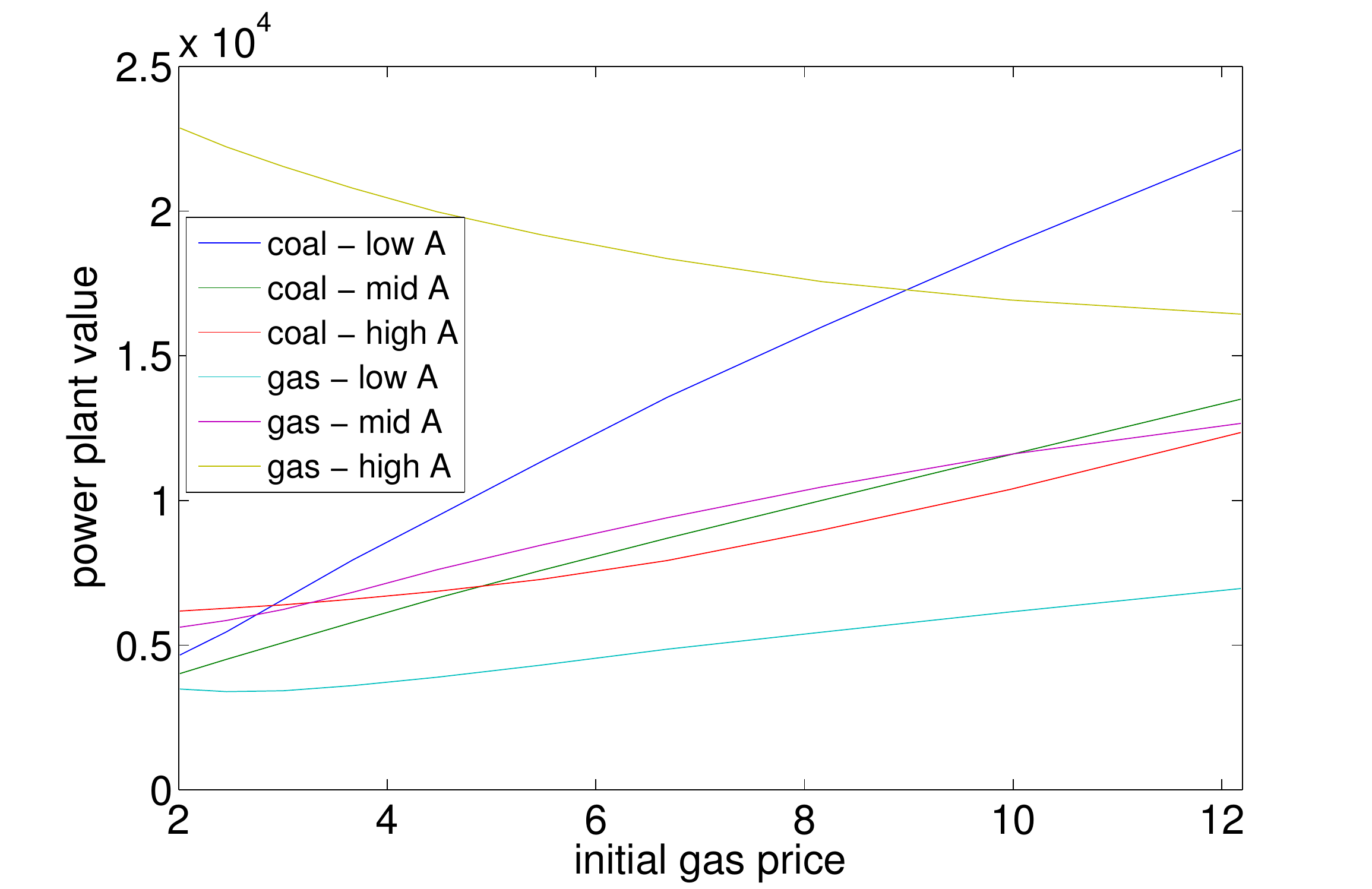}
 \hfill
 \includegraphics[width=0.45\textwidth]{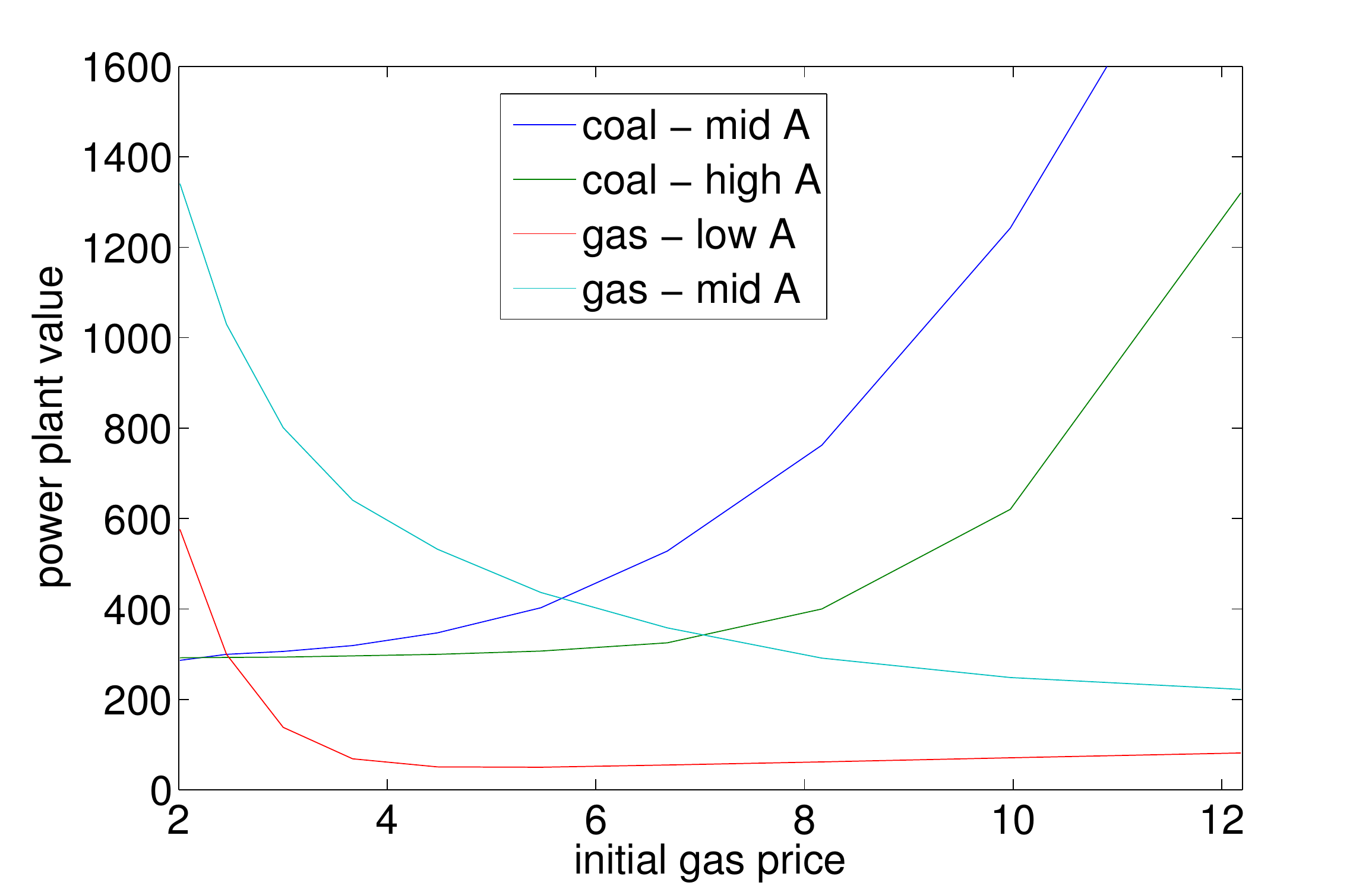}
  \caption{Power Plant Value (sum of spreads over $\tau$) versus $s_0^g$ for high efficiency (left) and low efficiency (right). `High A' corresponds to $\Gamma=1e+08$, `Mid A' to $\Gamma=1.4e+08$ (base case) and `Low A' to $\Gamma=1.8e+08$, with corresponding values $A_0=94,A_0=52,A_0=5$} 
 \label{fig:casestudy2}
\end{figure}
In Figure \ref{fig:casestudy2}, we plot the value of coal and gas power plants, as given by the sum of spread options of all maturities $\tau\in[0,T]$.  In the first plot, we consider high efficiency (low $h_v$ and $e_v$) plants, while in the second we consider low efficiencies.  The former are much more likely to operate each day and to generate profits, and are hence much more valuable than the latter.  However, they also show different relationships with $s^g_0$, as illustrated for several different cap levels $\Gamma$ (like in Case Study I above) which correspond to high, low or medium (base case) values of $A_0$.  Firstly, for low efficiency plants (right plot), we observe that gas plant value is typically decreasing in $s^g_0$, as we expect, since higher gas prices tend to push the bids from gas above those from coal, meaning there is less chance that the gas plant will be used for electricity generation.  Similarly, coal plant values are typically increasing in $s^c_0$, as more coal plants will be used.\footnote[2]{In this plot, the cases `coal - low A' and `gas - high A' are not included as their values are much greater and hence cannot be shown conveniently on the same axis.}  Note however, that for some cases, the curves flatten out, as no more merit order changes are possible.  This is particularly true for the coal plant when $A_0$ is very high (and hence once gas drops below a certain point, the coal plant is almost certainly going to remain more expensive than all gas plants) and for the gas plant when $A_0$ is very low (and hence once coal increases above a certain point, the gas plant is almost certainly going to remain more expensive than all coal plants). 

We now turn our attention to the high efficiency case (left plot), meaning the relatively cheap and clean plants for each technology.  As expected, coal benefits from low values of $A_0$ (ie, a lenient cap) and gas from high values of $A_0$ (ie, a strict cap).  On the other hand, the relationship with $s_0^g$ is now increasing for almost all six cases plotted except that of a gas plant with high $A_0$.  While it may seem surprising that for low or medium values of $A_0$, the gas plant value increase with $s^g_0$, this is quite intuitive when one considers that the range of bids from gas generators widens as $s^g_0$ increases, implying that the efficient plants can make a larger profit when the inefficient plants set the power price.  Indeed, as demand is quite high on average, and gas is 60 percent of the market, it is likely that these efficient gas plants will almost always be `in-the-money' even if coal is lower in the stack.  Only in the case that coal is typically above gas and now marginal (i.e. the high $A_0$ case) is the value of the gas plant decreasing in $s^g_0$ since the plant's profit margins shrink as gas and coal bids converge.

\subsection{Case Study III:  Comparison with Reduced-Form}

The second analysis we consider is to compare the results of our structural model for the allowance price, with two other simpler models, both of which belong to the class of `reduced-form' models.  The first of these treats the allowance price itself as a simple Geometric Brownian Motion (with drift $r$ under $\mathbb{Q}$), and hence $A_\tau$ is lognormal at spread maturity, like $S^c_\tau$ and $S^g_\tau$. The second comparison treats the emissions process as a Geometric Brownian Motion (GBM), and retains the digital terminal condition $A_T=\pi \mathbb{I}_{\{E_T>\Gamma\}}$.  As the drift of $(E_t)$ is then simply a constant (chosen to match the initial value $A_0$ in the full model), there is no feedback from $(A_t)$ on $(E_t)$, or in other words, no abatement induced by the allowance price.  For any time $t$, $A_t$ is then given in closed-form by a formula resembling the Black-Scholes digital option price.  In order to fully specify the reduced-form models, we need to choose volatility parameters $\sigma_a$ and $\sigma_e$ for each of the GBMs, as well as correlations $\rho_{ac},\rho_{ag}$ and $\rho_{ec},\rho_{eg}$ with the Brownian Motions driving the other exogenous factors, coal and gas prices.  All of these parameters are chosen to approximately match the levels of volatility and correlation produced by simulations in the full structural model, and are given in Table \ref{tab:numpar_GBMs}.  Finally, note that in all three models we compare, the power price is given by the same bid stack function as usual, so our aim is to isolate and evaluate the effect of our more sophisticated framework for the allowance price, in comparison to simpler approaches.  The cap throughout is $\Gamma=1.4e+08$, the base case.
\begin{table}
 \begin{center}
  \begin{tabular}{cccccc}
\hline
    $\sigma_a$ & $\rho_{ac}$ & $\rho_{ag}$ & $\sigma_e$ & $\rho_{ec}$ & $\rho_{eg}$ \\
\hline
     0.6 & -0.2  & 0.4 &  0.006 &   -0.2 & 0.2 \\
\hline
  \end{tabular}
 \end{center}
 \caption{Parameters for reduced-form comparisons, treating $A_t$ and $E_t$ as GBMs.}
 \label{tab:numpar_GBMs}
\end{table}

\begin{figure}
  \centering
  \includegraphics[width=\textwidth]{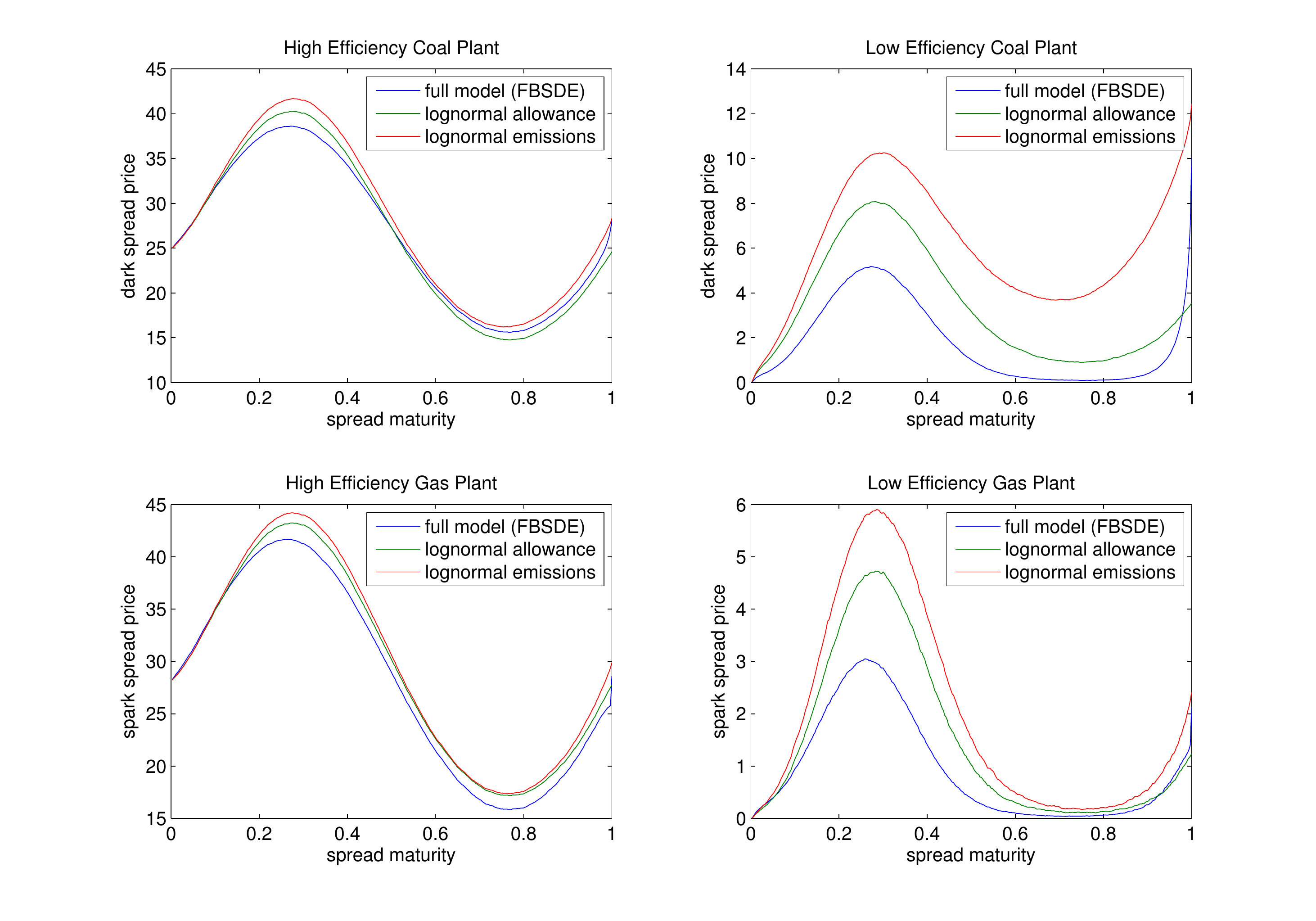}
  \caption{Model comparison against reduced-form:  Spark and dark spread option values for varying heat rates, emissions rates and maturities.} 
 \label{fig:casestudy3}
\end{figure}
Figure \ref{fig:casestudy3} reveals that the difference between the reduced form models and the full structural model is relatively small for high efficiency gas and coal plants which are typically `in-the-money'.  In contrast a larger gap appears for low efficiency cases, where the reduced form models significantly overprice spread options relative to the stack model.  In particular, the case of lognormal emissions produces much higher prices, especially for dark spreads.  The intuition is as follows.  In the full model, the bid stack structure automatically leads to lower emissions when the allowance price is high, and higher emissions when the allowance price is low, producing a mean-reversion-like effect on the cumulative emissions, keeping the process moving roughly towards the cap, with the final outcome (compliance or not) in many simulations only becoming clear very close to maturity.  In contrast, if $E_t$ is a GBM, much of the uncertainty is often resolved early in the trading period, with $A_t$ then sticking near zero or $\pi$ for much of the period.  In such cases, there is a much larger benefit for deep OTM options (low efficiency plants), for which the tails of the allowance price distribution provide great value either for coal (when the price is near zero) or for gas (when the price is near the penalty).  We observe that in some of the subplots (particularly low efficiency coal), this extra benefit is indeed realized in the full model, but only very near the end of the trading period when the volatility of $(A_t)$ spikes, and the process either rises or falls sharply.  In contrast, for the other reduced-form model with lognormal $(A_t)$, the volatility of the allowance price is constant throughout and $A_t$ never moves rapidly towards zero or the penalty. However, the overall link with fuel and power prices is much weaker when simply using correlated Brownian Motions, which serves to widen the spread distribution in most cases relative to the full structural model.  This result is somewhat similar to the observation in \cite{rCarmona2012} that a stack model generally produces lower spread option prices than Margrabe's formula for correlated lognormals.

\subsection{Case Study IV: Cap-and-Trade vs. Carbon Tax}\label{str:CATvsCT}

Finally, we wish to investigate the implications of the model for cap-and-trade systems, as compared with fixed carbon taxes.  This question has been much debated by policy makers as well as academics, and can be roughly summarized as fixing quantity versus fixing price.  In \cite{CarmonaFehrHinzPorchet}, several different designs for cap-and-trade systems are compared to a carbon tax, using criteria such as cost to society and windfall profits to power generators.  Here we follow a related approach by analyzing the power sector as a whole, but we build on our previous case studies by using spread option prices as a starting point.  Firstly we observe that the total expected discounted profits of the power sector are equal to the value of all the power plants implied by the bid stack structure, which in turn equals a portfolio of (or integral over) sums of spread option prices with varying $h_v$ and $e_v$. i.e, for each simulation over the period $[0,T]$, total profits (total revenues minus total costs) are\footnote[3]{Note that we do not consider here additional issues such as whether allowances are auctioned or freely allocated to generators.  Instead, we assume that allowances are bought on the market by generators as and when they need them.} 
\begin{eqnarray*}
\text{Total Profits } &=& \sum_{\tau \in [0,T]} \left( P_\tau D_\tau - \int_0^{D_\tau} b(x,A_\tau,S_\tau) \text{d}x \right) \\
&=& \sum_{\tau \in [0,T]} \int_0^{\bar{x}}\left(  P_\tau  - b(x,A_\tau,S_\tau) \right)^+\text{d}x  \\
&=& \sum_{\tau \in [0,T]} \left(\int_0^{\bar{x}_c}\left(  P_\tau  - h_c(x) S_\tau^c - e_c(x) A_\tau \right)^+\text{d}x + \int_0^{\bar{x}_g}\left(  P_\tau  - h_g(x) S_\tau^g - e_g(x) A_\tau \right)^+\text{d}x \right), 
\end{eqnarray*}
where the second line follows from the fact that the events $\{P_\tau \ge b(x,A_\tau,S_\tau)\}$ and $\{D_\tau \ge x\} $ are equal.

Hence, instead of picking particular coal and gas plants with efficiencies specified by the parameters in Table \ref{tab:numpar_spread}, we now integrate power plant value over all the efficiencies of plants in the stack, as defined by the parameters in Table \ref{tab:numpar_stack}.  For the case of the carbon tax, we simply force $A_t=A_0\exp(rt)$ for all $t\in[0,T]$, including the exponential function in order to match the mean of the process in the cap-and-trade model.  This is equivalent to setting the volatility $\sigma_a$ equal to zero in the GBM model for the allowance price in Case Study III. 

\begin{figure}
  \centering
\includegraphics[width=\textwidth]{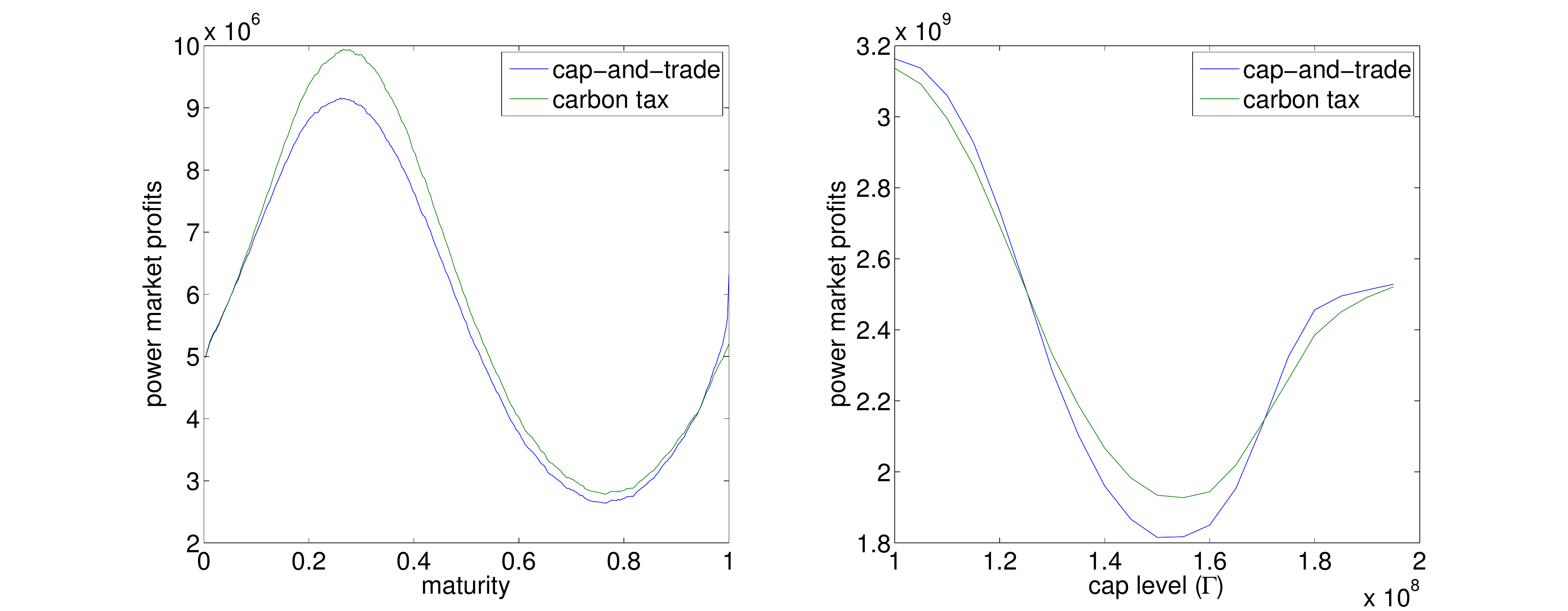}
\caption{Cap-and-trade vs. carbon tax: Power sector profits versus time for the `base case' (left);  Total profits over one year for equally-spaced cap values from $1e+08$ to $1.95e+08$ tons of CO$_2$. (right)} 
 \label{fig:casestudy4}
\end{figure}

In Figure \ref{fig:casestudy4}, we first plot the expected total market profits in the base case as a function of time.  It is interesting to observe that two important effects occur, pulling the profits in opposite directions, but varying in strength over the trading period.  In particular, although the profits must be equal at time zero, a gap quickly appears in the early part of the trading period, with expected profits to power generators significantly higher under a carbon tax than cap-and-trade. However, as maturity approaches, the gap narrows and the order reverses over the final days, as cap-and-trade generates higher expected profits.  These effects can be understood with a little thought.  Firstly, as $A_0=52$ in the base case, the bids of coal and gas begin the period at very similar levels, a state which generally keeps profits low, since the variance of electricity prices is low and the profit margins of both coal and gas generators are quite low.  As time progresses and fuel prices move, the coal and gas bids will tend to drift apart in most simulations, for example with gas sometimes moving above coal, say.  However, in our structural model for the cap-and-trade scheme, in such a case the higher emissions will induce a higher allowance price, and in turn a feedback effect due to the coupling in \eqref{eq:allowance_FBSDE}, which acts to keep coal and gas bids closer together.  A similar argument can be made for the case of gas bids tending to move below coal bids but then being counteracted by lower allowance prices.  Again we see that the power market structure induces mean reversion on $(E_t)$, which in this scenario (of an averagely strict cap) corresponds to keeping coal and gas bids close together.  On the other hand, under a carbon tax with fixed (or deterministic) $A_t$, there is of course no feedback mechanism (price-sensitive abatement), and bids tend to wander apart.  However, as the end of the trading period approaches, in the cap-and-trade system the allowance price eventually gets pulled to either zero or $\pi$, which will separate the bids in one way or the other, either leading to very large profits for coal plants (if $A_T=0$) or for gas plants (if $A_T=\pi$).  This is a similar effect to that discussed when comparing with a lognormal allowance price in Case Study III, as neither a carbon tax nor a lognormal allowance price model sees the extra volatility near maturity caused by the terminal condition.

Finally, in the second plot of Figure \ref{fig:casestudy4}, we consider how these conclusions change if the cap is made stricter or more lenient.  Instead of plotting against maturity, we consider the total profits of the power sector over the entire period $[0,T]$.  Firstly, we observe that under both forms of emissions regulation, power sector profits are lowest if the cap is chosen close to base case, under which the bids from coal and gas generators are more tightly clustered together.  Secondly, it is important to notice that the conclusion in the previous discussion that a carbon tax provides more profits to the power sector does not hold for all scenarios of the cap.  In particular, for either very high or very low values of the cap, the cap-and-trade scheme provides more profits than a tax.  The explanation here is that for the automatic abatement mechanism in the stack to have its largest impact (keeping bids together, and emissions heading towards the cap), there needs to be significant uncertainty at time zero as to whether the cap will be reached.  The feedback mechanism of a cap-and-trade system then allows this uncertainty to be prolonged through the period.  On the other hand, for an overly strict or overly lenient cap (or similarly for a merit order which does not allow for much abatement), the second effect discussed above dominates over the first.  In other words, the terminal condition which guarantees large profits to either coal or gas at maturity begins to take precedence earlier in the trading period, instead of just before maturity as in the base case.  Although in practice there are many other details to consider when comparing different forms of emissions legislation, our stylized single-period model sheds some light on the differences between cap-and-trade and carbon tax, as well as the clear importance of choosing an appropriate cap level. 
 
\section{Conclusion}

As policy makers debate the future of global carbon emissions legislation, the existing cap-and-trade schemes around the world have already significantly impacted the dynamics of electricity prices and the valuation of real assets, such as power plants, particularly under the well-known European Union Emissions Trading Scheme.  Together with the recent volatile behaviour of all energy prices (e.g., gas, coal, oil), the introduction of carbon markets has increased the risk of changes in the merit order of fuel types, known to be a crucial factor in the price setting mechanism of electricity markets.  In the US, the recent sharp drop in natural gas prices is already causing changes in the merit order, which would be further magnified by any new emissions regulation.  Such considerations are vital for describing the complex dependence structure between electricity, its input fuels, and emissions allowances, and thus highly relevant for both market participants and policy makers designing emissions trading schemes.  In this paper, we derived the equilibrium carbon allowance price as the solution of an FBSDE, in which feedback from allowance price on market emission rates is linked to the electricity stack structure.  The resulting model specifies simultaneously both electricity and allowance price dynamics as a function of fuel prices, demand and accumulated emissions; in this way, it captures consistently the highly state-dependent correlations between all the energy prices, which would not be achievable in a typical reduced-form approach. We used a PDE representation for the solution of the pricing FBSDE and implemented a finite difference scheme to solve for the price of carbon allowances.  Finally we compared our model for allowance prices with other reduced-form approaches and analysed its important implications on price behaviour, spread option pricing and the valuation of physical assets in electricity markets covered by emissions regulation.  The four case studies illustrated the many important considerations needed to understand the complex joint dynamics of electricity, emissions and fuels, as well as the additional insight that can be provided by our structural approach.  

\begin{appendix}
 
\section{Numerical Solution of the FBSDE}\label{ap:num_soln_FBSDE}

\subsection{Candidate Pricing PDE}
The construction of a solution to the FBSDE \ref{eq:allowance_FBSDE} was done in Theorem \ref{th:existence} by means of a decoupling random field $u$ representing the solution in the form $A_t=u(t,E_t)$. The existence of this random field was derived from the results of \cite{jMa2011}, and given its uniqueness and the Markov nature of
FBSDE \ref{eq:allowance_FBSDE}, it is possible to show that $u$ is in fact a function of $D_t$ and $S_t$, so that $A_t$ is in fact of the form $A_t=\alpha(t,D_t,E_t, S^c_t,S^g_t)$ for some deterministic function $\alpha:[0,T]\times[0,\bar{x}]\times\R_{++}\times[0,\bar{e}]\hookrightarrow [0,\pi]$. Standard arguments in the theory of FBSDEs show that this $\alpha$ is a viscosity solution of the semilinear PDE:
\begin{align}\label{eq:allowance_PDE}
 \mathcal{L}\alpha + \mathcal{N}\alpha &= 0, &&\text{on } U_T\\
 \alpha &= \phi(e), &&\text{on } \{t=T\}\times U,
\end{align}
where $U:=(0,\bar{x})\times\R_{++}\times\R_{++}\times(0,\bar{e})$ and $U_T:=[0,T)\times U$; the operators $\mathcal{L}$ and $\mathcal{N}$ are defined by
\begin{multline*}
 \mathcal{L}:= \frac{\partial }{\partial t} + \frac{1}{2}\sigma_d(d)^2\frac{\partial^2 }{\partial d^2} + \frac{1}{2}\sigma_c(s_c)^2\frac{\partial^2 }{\partial s_c^2} + \frac{1}{2}\sigma_g(s_g)^2\frac{\partial^2 }{\partial s_g^2}\\
 + \mu_d(t,d)\frac{\partial }{\partial d} + \mu_c(s_c)\frac{\partial }{\partial s_c} + \mu_g(s_g)\frac{\partial }{\partial s_g} - r\cdot
\end{multline*}
and $\mathcal{N}:= \mu_e(d,\cdot,(s_c,s_g))\frac{\partial }{\partial e}$. As previously, we specify for our purposes that  $\phi(e)= \pi I_{[\Gamma,\infty)}(e)$, for $e\in\R$.

With regards to the problem \eqref{eq:allowance_PDE} the question arises at which parts of the boundary we need to specify boundary conditions and, given the original stochastic problem \eqref{eq:allowance_FBSDE}, of what form these conditions should be. To answer the former question we consider the Fichera function $f$ at points of the boundary where one or more of the diffusion coefficients disappear (cf. \cite{oOleinik1973}). Defining $n:=(n_d,n_c,n_g,n_e)$ to be the inward normal vector to the boundary, Fichera's function for the operator $(\mathcal{N}+\mathcal{L})$ reads
\begin{multline}\label{eq:fichera}
 f(t,d,s_c,s_g,e) := \left(\mu_d-\frac{1}{2}\frac{\partial }{\partial d}\sigma_d^2\right)n_d + \left(\mu_c-\frac{1}{2}\frac{\partial }{\partial s_c}\sigma_c^2 - \frac{\partial }{\partial s_c}\rho \sigma_c\sigma_g\right)n_c\\
 + \left(\mu_g-\frac{1}{2}\frac{\partial }{\partial s_g}\sigma_g^2 - \frac{\partial }{\partial s_g}\rho \sigma_c\sigma_g\right)n_g + \mu_en_e, \quad \text{on } \partial U_T.
\end{multline}
At points of the boundary where $f\geq 0$ the direction of information propagation is outward and we do not need to specify any boundary conditions; at points where $f<0$ information is inward flowing and boundary conditions have to be specified. We evaluate \eqref{eq:fichera} for the choice of coefficients presented in \textsection \ref{str:stoch_proc}.

Considering the parts of the boundary corresponding to $d=0$ and $d=\bar{x}$, we find that $f\geq 0$ if and only if $\min(\bar{D}(t),\bar{x}-\bar{D}(t))\geq \bar{x}\hat{\sigma}$, which is the same condition prescribed in \textsection \ref{str:dem_proc} to guarantee that the Jacobi diffusion stays within the interval $(0,\bar{x})$. At points of the boundary corresponding to $e=0$, we find that $f\geq 0$ always. On the part of the boundary on which $e=\bar{e}$, $f< 0$ except at the point $(d,\cdot,\cdot,e)=(0,\cdot,\cdot,\bar{e})$, where $f=0$, an ambiguity which could be resolved by smoothing the domain. Similarly, we find that $f\geq 0$ on parts of the boundary where $s_c=0$ or $s_g=0$. Therefore, no boundary conditions are necessary except when $e=\bar{e}$, where we prescribe
\begin{equation}\label{eq:bc_emax}
\alpha=\exp(-r(T-t))\pi, \quad \text{on } U_T|_{e=\bar{e}}.
\end{equation}
In addition we need to specify an asymptotic condition for large values of $s_c$ and $s_g$. We choose to consider solutions that, for $i\in\{c,g\}$, satisfy
\begin{equation}\label{eq:bc_smax}
\frac{\partial \alpha}{\partial s_i}\sim 0, \quad \text{on } U_T|_{s_i\to \infty}.
\end{equation}

%

\subsection{An Implicit - Explicit Finite Difference Scheme}
\label{ap:FinDiff}
We approximate the domain $\bar{U}_T$ by a finite grid spanning $[0,T]\times[0,\bar{x}]\times[0,\bar{s}_c]\times[0,\bar{s}_g]\times[0,\bar{e}]$. For the discretization we choose mesh widths $\Delta d$, $\Delta s_c$, $\Delta s_g$, $\Delta e$ and a time step $\Delta t$. The discrete mesh points $(t_k,d_m,s_{c_i},s_{g_j},e_n)$ are then defined by
\begin{align*}
t_k &:= k\Delta t, & d_m &:= m\Delta d,\\
s_{c_i} &:= i\Delta s_c, & s_{g_j} &:= j\Delta s_g, & e_n &:= n\Delta e.
\end{align*}
The finite difference scheme we employ produces approximations $\alpha^k_{m,i,j,n}$, which are assumed to converge to the true solution $\alpha$ as the mesh width tends to zero.

Since the partial differential equation \eqref{eq:allowance_PDE} is posed backwards in time with a terminal condition, we choose a backward finite difference for the time derivative. In order to achieve better stability properties we make the part of the scheme relating to the linear operator $\mathcal{L}$ implicit; the part relating to the operator $\mathcal{N}$ is made explicit in order to handle the nonlinearity. 

In the $e$-direction we are approximating a conservation law PDE with discontinuous terminal condition. (For an in depth discussion of numerical schemes for this type of equation see \cite{rLeVeque1990}) The first derivative in the $s$-direction, relating to the nonlinear part of the partial differential equation, is discretised against the drift direction using a one-sided upwind difference. Because characteristic information is propagating in the direction of decreasing $e$, this one-sided difference is also used to calculate the value of the approximation on the part of the boundary corresponding to $e=0$. At the part of the boundary corresponding to $e=\bar{e}$ we apply the condition \eqref{eq:bc_emax}.

In the $d$-direction the equation is elliptic everywhere except on the boundary, where it degenerates. Therefore, we expect the convection coefficient to be much larger than the diffusion coefficient near the boundaries. In order to keep the discrete maximum principle we again use a one-sided upwind difference for the first order derivative. Thereby we have to pay attention that due to the mean-reverting nature of $(D_t)$ the direction of information propagation and therefore the upwind direction changes as the sign of $\mu_d$ changes. The same upwind difference is also used to calulate the value of the approximation at the boundaries $d=0$ and $d=\bar{x}$. To discretize the second order derivative we use central differences.

The $s_c$ and $s_g$-direction are treated similarly to the $d$-direction. We use one-sided upwind differences for the first order derivatives, thereby taking care of the boundaries corresponding to $s_c=0$ and $s_g=0$. The second order derivatives are discretized using central differences. At the boundary corresponding to $s_c=\bar{s}_c$ and $s_g=\bar{s}_g$ we apply the asymptotic condition \eqref{eq:bc_smax} as a boundary condition.

With smooth boundary data, on a smooth domain, the scheme described above can be expected to exhibit first order convergence. In our setting, we expect the discontinuous terminal condition to have adverse effects on the convergence rate.

\section{Numerical Calculation of Spread Prices}\label{ap:num_soln_spread}

\subsection{Time Discretisation of SDEs}
Let $(D_k,S^c_k,S^g_k,E_k,A_k)$ denote the discrete time approximation to the FBSDE solution $(D_t,S^c_t,S^g_t,E_t,A_t)$ on the time grid $0 < \Delta t < 2\Delta t < \ldots < n_k\Delta t = \tau$. At each time step we calculate $A_k$ by interpolating the discrete approximation $\alpha^k_{m,i,j,n}$ at $(D_k,S^c_k,S^g_k,E_k)$, beginning with the initial values $D_0=d_0, S^c_0 = s^c_0, S^g_0 = s^g_0, E_0=0$. The approximations $(D_k,S^c_k,S^g_k,E_k)$ are obtained using a simple Euler scheme (cf. \cite{pGlasserman2004}). The discretized version of $(D_t)$ is forced to be instantaneously reflecting at the boundaries $D_k=0$ and $D_k=\bar{x}$; similarly, the discretized versions of $(S^c_t)$ and $(S^g_t)$ are made instantaneously reflecting at $S^c_k=0$ and $S^g_k=0$.

\subsection{Monte Carlo Calculation of Option Prices}
Using this discretization we simulate $n_{mc}$ paths and, as usual, for $t\in[0,\tau)$, calculate the mean spark spread price $\hat{V}_t$, given by
\begin{equation*}
\hat{V}_t:=\exp(-r(\tau-t))\frac{1}{n_{mc}}\sum_{i=1}^{n_{mc}} \left(b(D^i_{n_k},S^{c,i}_{n_k},S^{g,i}_{n_k},A^i_{n_k})-h_vS^{g,i}_{n_k}- e_vA^i_{n_k}\right)^+,
\end{equation*}
where the index $i$ refers to the simulation scenario. The corresponding standard error $\hat{\sigma}_{v}$ is obtained by
\begin{equation*}
\hat{\sigma}_{v}:= \sqrt{\frac{1}{n_{mc}\left(n_{mc}-1\right)}\sum_{i=1}^{n_{mc}}\left(V^i_{n_k}-\hat{V}_{\tau}\right)^2}.
\end{equation*}
\end{appendix}

\addcontentsline{toc}{chapter}{Bibliography}
\renewcommand{\bibname}{References}

\bibliography{CleanSpread_References.bib}
\bibliographystyle{apalike}
\end{document}